%% file: final.tex
\documentclass[conference]{IEEEtran}
\IEEEoverridecommandlockouts
\usepackage{cite}
\usepackage{amsmath,amssymb,amsfonts}
\usepackage{textcomp}
\usepackage{xcolor}
\def\BibTeX{{\rm B\kern-.05em{\sc i\kern-.025em b}\kern-.08em
    T\kern-.1667em\lower.7ex\hbox{E}\kern-.125emX}}

\usepackage{amssymb}
\usepackage{comment}
\usepackage{mathtools} 
\usepackage{amsthm}
\usepackage{amsmath}
\usepackage{amsfonts} 
\usepackage{graphicx}
\usepackage{float}
\usepackage{balance}

\usepackage{xcolor, soul}
\sethlcolor{yellow}

\usepackage{caption}
\usepackage{subcaption}

\usepackage{algorithm}

\usepackage[noend]{algpseudocode}
\newcommand\numberthis{\addtocounter{equation}{1}\tag{\theequation}}

\newcommand{\rom}[1]{\uppercase\expandafter{\romannumeral #1\relax}}

\newtheorem{remark}{Remark}
\newtheorem{corollary}{Corollary}
\newtheorem{proposition}{Proposition}

\input{defns}

\allowdisplaybreaks
    
\begin{document}

\title{Error Probability Bounds for Invariant Causal Prediction via Multiple Access Channels}

\author{\IEEEauthorblockN{Austin Goddard and Yu Xiang}
\IEEEauthorblockA{\textit{Department of Electrical and Computer Engineering} \\
\textit{University of Utah}\\
\{austin.goddard,\,yu.xiang\}@utah.edu}
\and
\IEEEauthorblockN{Ilya Soloveychik}
\IEEEauthorblockA{\textit{Department of Statistics} \\
\textit{Hebrew University of Jerusalem}\\
soloveychik.ilya@mail.huji.ac.il}
}

\maketitle

\begin{abstract}

We consider the problem of lower bounding the error probability under the invariant causal prediction (ICP) framework. To this end, we examine and draw connections between ICP and the zero-rate Gaussian multiple access channel by first proposing a variant of the original invariant prediction assumption, and then considering a special case of the Gaussian multiple access channel where a codebook is shared between an \emph{unknown} number of senders. This connection allows us to develop three types of lower bounds on the error probability, each with different assumptions and constraints, leveraging techniques for multiple access channels. The proposed bounds are evaluated with respect to existing causal discovery methods as well as a proposed heuristic method based on minimum distance decoding. 
\end{abstract}


\section{Introduction}
\noindent The recent invariant causal prediction (\textsf{ICP}) framework~\cite{peters2016causal} has pioneered the study of leveraging invariance in datasets across different experimental settings (or environments) for identifying potential causal predictors (with various theoretical extensions and applications~\cite{heinze2018invariant,pfister2019invariant,meinshausen2016methods,goddard2022invariance}). For linear models, the underlying assumption~\cite[Assumption~1]{peters2016causal} requires the existence of invariant (w.r.t. environments) coefficients as well as noise variables. In this work, we focus on developing lower bounds on the error probability of \textsf{ICP} by making a connection between it and the Gaussian multiple access channel (\textsf{MAC}). To do so, we make the following assumption on invariance: there exists a vector of coefficients $\g^* = [\g^*_1,\dots,\g^*_m]^\top\in \mathbb{R}^{m}$ with support $S^* \coloneqq \{i:\gamma^*_i \neq 0 \}$ that satisfies 
 \begin{equation}\label{equ:inv_ass1}
  \text{For all } e \in \Ec: Y^e = X^e\g^* + N^e, N^e \sim \mathcal{N}(0,\s^2I),
 \end{equation}
 where $\Ec$ denotes the set of different environmental conditions, $\sigma^2$ is \emph{unknown} but in a \emph{known} range $[\sigma^2_\text{min}, \sigma^2_\text{max}]$, and $X^e = [x^e_1,\dots,x^e_m] \in \mathbb{R}^{n_e\times m}$ represents a \emph{deterministic} design matrix with $n_e$ denoting the number of samples in environment $e$. We shall see (in Section~\ref{sec:bounds}) we only need $\sigma^2_{\text{min}}$ for the bounds; thus for our purpose, it is equivalent to assume that $\sigma^2$ is known. Our main assumption stated in~\eqref{equ:inv_ass1} differs from~\cite[Assumption~1]{peters2016causal}. We argue it captures the essence of the original while allowing us to leverage techniques from Gaussian \textsf{MAC} settings (see Section~\ref{sec:diff} for details).

Bounds on the probability of error for a Gaussian \textsf{MAC} with a shared codebook exist for the positive-rate case (positive-rate meaning that $m$ is exponential in $n$)~\cite{jin2011limits}. This, however, is not the main focus of this work. Rather, we focus on the zero-rate case where $m$ does not grow with $n$. To the best of our knowledge, we know of no existing lower bounds for the zero-rate and multiple sender case when the Gaussian MAC has a shared codebook and an unknown number of senders. When there is one sender, it is not necessary to consider the codebook shared and current Gaussian point-to-point channel results can be used. To this end, a bound on the probability of error for \textsf{ICP} has been previously proposed in~\cite{goddard2022lower}. This bound is lacking in two main aspects: (1) it requires the number of environments to be two, and (2) it can only handle the single sender case (i.e., $|S^*| = 1$). The bounds proposed in this work are more general in that they apply to both an arbitrary number of senders and an arbitrary number of environments.


\section{Problem Formulation} \label{sec:ICPandMAC}
We now formally describe the problem. Consider a setting in which a vector of coefficients $\g^*\in\mathbb{R}^{m}$ is generated such that its support $S^*\subseteq \{1,\dots,m\}$ is \emph{uniformly} chosen from all subsets of $\{1,\dots,m\}$. Let $w\in\mathbb{R}^{m}$ be a fixed vector where each element $w_i$ is non-zero.  The generation of $\g^* = [\g^*_1, \dots, \g^*_m]^\top$ then follows from $\g^*_{i}=w_i$ if $i \in S^*$ and $\g^*_{i}=0$ if $i \not\in S^*$, for every $i \in \{1,\dots,m\}$. The number of non-zero coefficients in $\g^*$ is referred to as $k=|S^*|$. 


Let $X^e_S$ denote the matrix containing only the columns of $X^e$ indexed by the set $S$. The columns in the matrix $X_{S^*}^e$ are referred to as causal predictors. It is important to note the distribution of $Y^e$ in~\eqref{equ:inv_ass1} is unique for any given $\g^*$, implying that the support $S^*$ is recoverable for each environment (see Remark~\ref{rem} below for considerations regarding collision). 

Upon receiving $Y^e$ for each $e \in \Ec$, one wishes to recover the support $S^*$  (or equivalently $\g^*$). Let $\hat{S}^*_e$ be the estimate of $S^*$ for some $e\in\Ec$. We correctly recover the support $S^*$ if $\hat{S}^*_e = S^*$ for \emph{all} $e$. In the case where $\hat{S}^*_e\ne S^*$ for any $e\in\Ec$, the recovered $\g^*$ will not be invariant over all environments as in~\eqref{equ:inv_ass1}. Thus, the probability of error in recovering $S^*$ is
\begin{equation}\label{eqn::perr1}
    P_{\text{err}} = \Ps\{ \hat{S}^*_e\ne S^* \text{ for any } e\in\Ec\}, 
\end{equation}
where the probability is taken over the two sources of randomness, i.e., the random support $S^*$ and noise $N^e$.

\vspace{-0.5em}
\begin{remark}
\label{rem}
For simplicity of presentation, we assume $X^e$ and $w$ exist such that no collisions occur, namely, $X^e\g^1 \ne X^e\g^2$, for any $\g^1\ne\g^2$ generated as described above (this holds almost surely for any continuous random design matrix $X^e$). 
\end{remark}
\vspace{-0.5em}

\subsection{Differences between Assumptions}
\label{sec:diff}
 In an effort to derive bounds on $P_{\text{err}}$ in~\eqref{eqn::perr1}, our main assumption in~\eqref{equ:inv_ass1} differs from~\cite[Assumption~1]{peters2016causal} in several subtle ways that we now make an attempt to justify.

 \noindent\underline{\bf Coefficients (known vs. unknown):} The coefficients $w$ are known, whereas in~\cite[Assumption~1]{peters2016causal}, they need to be estimated. Since we focus on the lower bounds, the known coefficients serve as an oracle setting for our purpose, as knowing the coefficients can only bring the error probability lower. Furthermore, we develop algorithms for unknown coefficient settings (Algorithm~$2$). Note that the locations of the zero entries in $\g^*$ are unknown for both settings.

  \noindent\underline{\bf Noise distribution (arbitrary vs. Gaussian):} In~\cite[Assumption~1]{peters2016causal}, the noise distribution is zero-mean but otherwise arbitrary, while their methods and identification results rely on Gaussian noise. We fix the noise distribution in~\eqref{equ:inv_ass1} to be the widely adopted Gaussian noise, while keeping the variance to be unknown but within a known range; even though we shall see that for deriving the lower bounds, it is equivalent to assuming that $\sigma^2$ is known, as shown in our bounds. 

 \noindent\underline{\bf Design matrix (random vs. deterministic):} While the predictor variables $X^e$ are modeled to be random in~\cite[Assumption~1]{peters2016causal}, we focus on the deterministic $X^e$ since we focus on lower bounds. One extreme case is when the design matrix has \iid entries, which has been adopted in the information-theoretic perspective of compressive sensing (e.g.,~\cite{jin2011limits,scarlett2016limits}). However, this is not an interesting setting for \textsf{ICP} as the underlying causal predictors (i.e., columns of the design matrix) are likely to be generated from some joint distributions rather than a product distribution (i.e., the independent setting). Since our goal is to derive lower bounds on the error probability, it is then reasonable to consider deterministic settings to model the ``best" possible environment $X^e$ and characterize the corresponding error probability.

Unique recovery of $S^*$ for each environment in~\eqref{equ:inv_ass1} allows us to define the probability of error as in~\eqref{eqn::perr1}. The original formulation of \textsf{ICP} in~\cite{peters2016causal} is more general.  Without the assumptions mentioned above, a unique recovery of $S^*$ is not guaranteed. As a result, the best that can be done in~\cite{peters2016causal} is to focus on the intersection of all possible $S^*$, which leads to stronger guarantees under the family-wise error rate (FWER) but can oftentimes be conservative by reporting no discoveries. Also, \textsf{ICP} requires that the distribution of the noise and model coefficients not change over environment (i.e., $Y$ cannot be intervened), which can be relaxed in~\cite{du2022learning,du2023generalized,du2023identifying}.

\subsection{\textsf{ICP} vs. Gaussian \textsf{MAC}}
We now examine the Gaussian \textsf{MAC} and its connection with our proposed variant of \textsf{ICP} in~\eqref{equ:inv_ass1}. In such a channel, each sender $i\in \{1,\dots k \}$ has access to a codebook $C^i = \{c^i_1,c^i_2,\dots,c^i_m\}$, where $c^i_j \in \mathbb{R}^n$ and $m$ is the number of codewords in $C^i$. To transmit information, the $i$th sender first chooses a codeword and then sends the $t$-th element of the chosen codeword at transmission time $t$ as the input symbol $x_{i,t}$ such that the receiver obtains 
\begin{equation}\label{eqn::mac_mdl}
  y_t = h_1x_{1,t} + h_2x_{2,t} + \dots + h_kx_{k,t} +N_t,
\end{equation}
where $h_i$ is the channel gain for sender $i$, and $N_t \sim \mathcal{N}(0,\s^2)$ for all $t\in \{1\dots n \}$. Conventionally, the number of senders $k$ is known. The Gaussian \textsf{MAC} of interest here is one in which no collisions occur.  Additionally, it is often assumed there exists a \emph{total power constraint} over all codewords such that $\sum_{i=1}^m \sum_{t=1}^n x_{i,t}^2  \leq nmP$.
After $n$ transmissions, the receiver needs to determine which codewords in codebook $C$ were sent. 

While models in~\eqref{equ:inv_ass1} and~\eqref{eqn::mac_mdl} appear similar, we list several important differences between our formulation of \textsf{ICP} and codeword recovery in the Gaussian \textsf{MAC}. 
\begin{enumerate}
    \setlength\itemsep{0em}
    \item In \textsf{ICP}, there is a single shared ``codebook'' $X^e$ whereas in a Gaussian \textsf{MAC}, each sender has its own codebook. 
    \item In \textsf{ICP}, $k$ is arbitrary (i.e., some unknown value in $\{0,
    \dots,m\}$). Conventionally, in a Gaussian \textsf{MAC}, $k$ is assumed to be fixed. 
    \item There is no notion of \emph{environment} in the Gaussian \textsf{MAC}. 
\end{enumerate}
These two models are bridged by (a) considering a codebook shared by all senders, and (b) by assuming the number of senders $k$ is arbitrary. 
Additionally, if the number of environments is one, the problem of recovering $\g^*$ reduces to recovery in a Gaussian \textsf{MAC} with a shared codebook and arbitrary $k$. Indeed, recovering $S^*$ for an individual environment in~\eqref{equ:inv_ass1} is identical to codeword recovery in a Gaussian \textsf{MAC}. We later leverage this observation to identify limits on $P_{\text{err}}$ in~\eqref{eqn::perr1}.

\section{Lower Bounds on Error Probability}
\label{sec:bounds}
We provide three bounds on $P_{\text{err}}$ in~\eqref{eqn::perr1}. Each bound is derived based on the observation that individual environments can be treated as a Gaussian \textsf{MAC} with a shared codebook and unknown $k\in\{0,
\dots,m\}$. Because of this, $P_{\text{err}}$ is at least as small as the largest probability of error for codeword recovery in a Gaussian \textsf{MAC} over all $e\in\Ec$. Specifically, 
\begin{align}
P_{\text{err}} &= \Ps\biggl\{ \bigcup_{e\in\Ec} \{ \hat{S}^*_e \ne S^* \}\biggr\} 
 \ge \max_{e\in\Ec} \Ps\{\hat{S}^*_e \ne S^*\}, \label{eqn:bnd_prt1}
\end{align}
where $\Ps\{\hat{S}^*_e \ne S^*\}$ is the probability of error in recovering a $S^*$ in a Gaussian \textsf{MAC} with shared codebook and arbitrary $k$. From here, we leverage existing multiple-access results and strategies to bound $\Ps\{\hat{S}^*_e \ne S^*\}$. 

All else being equal, the probability of error for a random $\s^2$ between $[\sigma^2_\text{min}, \sigma^2_\text{max}]$ will never go lower than it will when $\s^2 = \s^2_\text{min}$. Consequently, $P_{err}$ for the random $\s^2$ setting is bounded by the setting in which the variance is fixed to be $\s^2_\text{min}$. We bound in this manner by treating the Gaussian \textsf{MAC} in~\eqref{eqn:bnd_prt1} as one having a noise variance $\s_\text{min}^2$.

Each proposed bound differs in the assumptions and the constraints used. Two bounds proposed assume a constraint on $X^e$ while one does not. The bounds dependent on constraints apply to any $X^e$ satisfying the associated constraint, and, as such, are looser than the bound applying to a specific $X^e$.

Let the quantity $v^e_{S} \in \mathbb{R}^{n_e}$ be the sent signal if codewords indexed by the set $S$ were sent. The sent signal at $t \in \{1,\dots, n_e \}$ is then $v_{S,t}^e = \sum_{i \in S} w_ix^e_{i,t}$. 
We define $\Tc_m$ to be the set of all subsets of $\{1,\dots,m\}$. For some $S,S'\in\Tc_m$, let the euclidan distance between $v^e_{S}$ and $v^e_{S'}$ be $d^e_{S,S'} =  || v^e_{S} - v^e_{S'}||_2$. We denote the Gaussian cumulative distribution function as $\Phi(\cdot)$.

\smallskip
\noindent\underline{\bf Bound I (data dependent bound):}
We first propose a bound on $P_{\text{err}}$ in which there is no power constraint. The bound is derived by determining the probability that the sent signal $v^e_{S^*}$ is incorrectly decoded using the distance between $v^e_{S^*}$ and the next closest signal. 
\begin{proposition} \label{cor:lb1}
    The probability of error is lower bounded by 
    \begin{equation} \label{eqn:lb1}
    P_{\text{err}} \ge \max_{e\in\Ec} \frac{1}{2^m} \sum_{S\in\Tc_m}  \Phi\left( -\frac{1}{2\s_{\text{min}}} \min_{S'\in\Tc_m\setminus S} d^e_{S,S'} \right).
    \end{equation}
\end{proposition}
\begin{proof}
Let $\Ps\{\hat{S}^*_e \ne S^*\}$ be the probability of incorrectly decoding $S^*$ in a Gaussian \textsf{MAC} with a shared codebook and arbitrary $k$. As $S^*$ is chosen uniformly from $\Tc_m$, 
 \begin{equation*}
\Ps\{\hat{S}^*_e \ne S^*\} = \frac{1}{2^m}\sum_{S\in\Tc_m} \Ps\{ \hat{S}^*_e \ne S^* | S^* = S\}. \label{eqn::perr_bnd_3}
\end{equation*}
Since Gaussian noise is symmetric, the probability $S^*$ is not recovered for some environment $e\in\Ec$ is at least the probability the target $Y^e$ is half the distance to the next closest possible sent signal. That is,
\begin{equation}
 \Ps\{ \hat{S}^*_e \ne S^* | S^* = S\} \ge \Phi\left( -\frac{1}{2\s_{\text{min}}} \min_{S'\in\Tc_m\setminus S} d^e_{S,S'} \right).
\end{equation}
The probability of error in~\eqref{eqn:lb1} then follows from~\eqref{eqn:bnd_prt1}. 
\end{proof}


\smallskip
\noindent\underline{\bf Bound II (under power constraint):} An alternative bound can be derived by considering a total power constraint of
\begin{equation}\label{eqn::peak_pwr_2}
    \sum_{i=1}^m \sum_{t=1}^{n_e} (w_ix_{i,t}^e)^2  \leq mn_eP_e.
\end{equation}
With this, we derive an upper bound on the average $d^e_{S,S'}$ over all combinations of $S$ and $S'$, and follow the same minimum distance argument used by Shannon in~\cite{shannon1959probability} to derive the bound in Proposition~\ref{cor:lb2} (see Appendix). Additionally, we present a second bound using~\eqref{eqn::peak_pwr_2} in Corollary~\ref{cor:cor1}. While the bound in Proposition~\ref{cor:lb2} is tighter, it can be simplified in Corollary~\ref{cor:cor1}. We note that this result can be of independent interest (e.g., the channel with feedback under the zero-rate setting~\cite{xiang2013gaussian}).
\begin{proposition} \label{cor:lb2}
 Let $X^e$ obey the power constraint in~\eqref{eqn::peak_pwr_2} for each $e \in \mathcal{E}$. The probability of error in recovering $\g^*$ is then lower bounded by 
    \begin{equation} \label{eqn:bnd2}
    P_{\text{err}} \ge \max_{e\in\Ec}\sum_{i = 1}^{m-1}\frac{1}{2^i}  \Phi\left( - \sqrt{\frac{2^{m-i}(m-i)n_eP_e}{4\s^2_{\text{min}}\left(2^{m-i}-1\right)}} \right).
    \end{equation}
\end{proposition}
\begin{corollary}\label{cor:cor1}
Under the same conditions as in~Proposition~\ref{cor:lb2},
\begin{equation} \label{eqn:cor1}
    P_{\text{err}} \ge (1-2^{-m/2}) \max_{e\in\Ec}\Phi\left( - \sqrt{\frac{2^{m/2}mn_eP_e}{8\s^2_{\text{min}}\left(2^{m/2}-1\right)}} \right).
\end{equation}
\end{corollary}
\begin{proof}
 First, \eqref{eqn:bnd2} can be further bounded by taking only the first $m/2$ terms of the sum. Then, as the remaining terms are monotonically decreasing in $i$, we replace each term with the last to derive~\eqref{eqn:cor1}.
\end{proof}

\smallskip
\noindent\underline{\bf Bound III (under a variant of power constraint):} We now consider a constraint such that 
\begin{equation} \label{eqn:sig_cnst_env}
    \sum_{t=1}^{n_e} \sum_{S\in\Tc_m} (v^e_{S,t})^2 \leq 2^mn_eQ_e,
\end{equation}
for all $e\in\Ec$. That is, instead of considering a constraint on the codewords, we now constrain all possible sent signals. With this, we view the problem as a conventional Gaussian point-to-point channel with a power constraint as in~\eqref{eqn:sig_cnst_env} and leverage existing Gaussian point-to-point channel results to bound $P_{\text{err}}$. As was also done in Corollary~\ref{cor:cor1}, we further bound Proposition~\ref{cor:lb3} to present a more simplified expression in Corollary~\ref{cor:cor2} by directly following~\cite{shannon1959probability}.
\begin{proposition} \label{cor:lb3}
     Let the predictors $X^e$ obey the constraint in~\eqref{eqn:sig_cnst_env} for each $e \in \mathcal{E}$. The probability of error in recovering $\g^*$ is then lower bounded by
    \begin{equation}\label{eqn:lb3}
        P_{\text{err}} \ge \max_{e\in\Ec}\frac{1}{2^m} \sum_{i = 1}^{2^m-1}  \Phi\left( - \sqrt{\frac{(2^{m}-i)n_eQ_e}{2\s^2_{\text{min}}\left(2^m-1-i\right)}} \right).
    \end{equation}
\end{proposition}
\begin{proof}
 In a change of perspective, treat the collection of sendable signals $v^e_{S,t}$ for $S\in\Tc_m$ as a codebook in a Gaussian point-to-point channel. The peak energy constraint is then given by~\eqref{eqn:sig_cnst_env}. A bound on the  probability $\Ps\{ \hat{S}^*_e \ne S^*\}$ then follows directly from~\cite{shannon1959probability} and thus~\eqref{eqn:lb3}. 
\end{proof}

\begin{corollary} \label{cor:cor2}
Under the same conditions as in~Proposition~\ref{cor:lb3},
\begin{equation} \label{eqn:cor2}
    P_{\text{err}} \ge \max_{e\in\Ec}\frac{1}{2}\Phi\left( - \sqrt{\frac{2^{m}n_eQ_e}{2\s^2_{\text{min}}\left(2^{m}-2\right)}} \right).
\end{equation}
\end{corollary}

\section{Algorithms and Experiments}
\subsection{Algoirthms}\label{ssec:methods}
We provide two heuristic methods with which to compare the bounds proposed. The first is an adaptation of Method~II presented in~\cite{peters2016causal}. Simply, Method II iterates over all subsets of variables. It fits a linear model to the data and tests the invariance of the residuals over all environments. The recovered support is the intersection of all invariant subsets (see~\cite{peters2016causal}). To fit our setting, we make the following changes. 
\begin{enumerate}
    \setlength\itemsep{0em}
    \item The coefficients are no longer estimated but are known.
    \item The mean and variance of the residuals are compared directly to their known values.
    \item As $S^*$ is unique, no intersection needs to be calculated. Rather, the recovered support is the subset deemed the ``most invariant'' (based on the largest p-value). 
\end{enumerate} 
We refer to this adaption of Method~II as  \textsf{MII\_known}.

Additionally, we propose a simple alternative that is a natural extension to the minimum distance decoding (\textsf{MDD}) algorithm for use in multiple environments. We refer to this method as \textsf{ICP\_MDD\_known}. An outline of this method is as follows. For each environment, the distance between the received signal $Y^e$ and $v^e_{S}$ for all $S\in\Tc_m$ is calculated. The recovered support $\Sh^e$ is the $S\in\Tc_m$ corresponding to the $v^e_{S}$ closest to $Y^e$ for all $e\in\Ec$. If the same $\Sh^e$ is recovered over all environments, we take this as the final recovered support. See Algorithm~\ref{alg:ICP_MDD} for a more detailed explanation. We also note that a current obstacle for those implementing ICP is that most methods have exponential complexity. Both Method~II and \textsf{ICP\_MDD\_known} are no exception. 

\begin{algorithm}
\caption{\textsf{ICP\_MDD\_known}}\label{alg:ICP_MDD}
\hspace*{\algorithmicindent} \textbf{Input:} Response $Y^e$, predictors $X^e$, and coefficients $w$

\hspace*{\algorithmicindent}
\textbf{Output:} $\hat{S}^*$ or nothing if no invariance was identified
\begin{algorithmic}
\For{every $e\in\Ec$}
\State $\Sh^e \coloneqq \arg\min_{S\in \Tc_m} ||Y^e-\sum_{i\in S} w_iX_i^e||_2$
\EndFor

\If{$\Sh^e$ is identical over all $e$} 
    \State \textbf{return }$\Sh^* \coloneqq \Sh^e$
\Else 
    \State \textbf{return} nothing
\EndIf 
\end{algorithmic}
\end{algorithm}

\vspace{-1.5em}
\subsection{Simulations}
\noindent\underline{\bf Simplex codes}.
We first compare the proposed bounds in a setting where the predictors $X^e$ constitute simplex codes. It is well known that simplex codes paired with \textsf{MDD} provide optimal recovery for a Gaussian point-to-point channel in the zero-rate setting. It is, however, important to note that simplex codes are not necessarily optimal for codeword recovery in a Gaussian \textsf{MAC} with unknown $k$ and a shared codebook. Further research is needed to examine the optimal codes in this special setting. Nonetheless, we expect the empirical results for such codes to be closer to the bounds when compared to other randomly generated codes.

We examine the setting where $m=3$. The matrix $X^e$ consists of simplex codes on a sphere or radius $\sqrt{n_e}$ such that  $x^e_i = \sqrt{n_e}\cdot[a_i,0,0,\cdots,0\;]^\top$,
where $a_1 = [1,0]$, $a_2 = [-\frac{1}{2},-\frac{\sqrt{3}}{2}]$, and $a_3 = [-\frac{1}{2},\frac{\sqrt{3}}{2}]$ for each $e\in\Ec$. As the columns of $X^e$ lie inside a sphere of radius $\sqrt{n_e}$, the constraints in~\eqref{eqn::peak_pwr_2} and~\eqref{eqn:sig_cnst_env} are satisfied using $P_e = 1$ and $Q_e = \frac{3}{4}$. The response $Y^e$ is generated such that $Y^e = X^e + N$, where $N\sim\mathcal{N}(0,I)$. We report the empirical probability of error averaged over $1000$ simulated datasets. 

The simulation results indicate both \textsf{ICP\_MDD\_known} and \textsf{MII\_known} approach the bound in Proposition~\ref{cor:lb1} with \textsf{ICP\_MDD\_known} performing slightly better for smaller sample sizes (Figure~\ref{fig:simp_ss}). Seeing that the empirical results remained some distance away from the bounds in Propositions~\ref{cor:lb2} and~\ref{cor:lb3} further justifies the remark that simplex codes are a sub-optimal coding scheme for this \textsf{MAC} setting. 

\begin{center}
\begin{figure*}
  \begin{subfigure}[b]{0.333\textwidth}    \includegraphics[width=\linewidth]{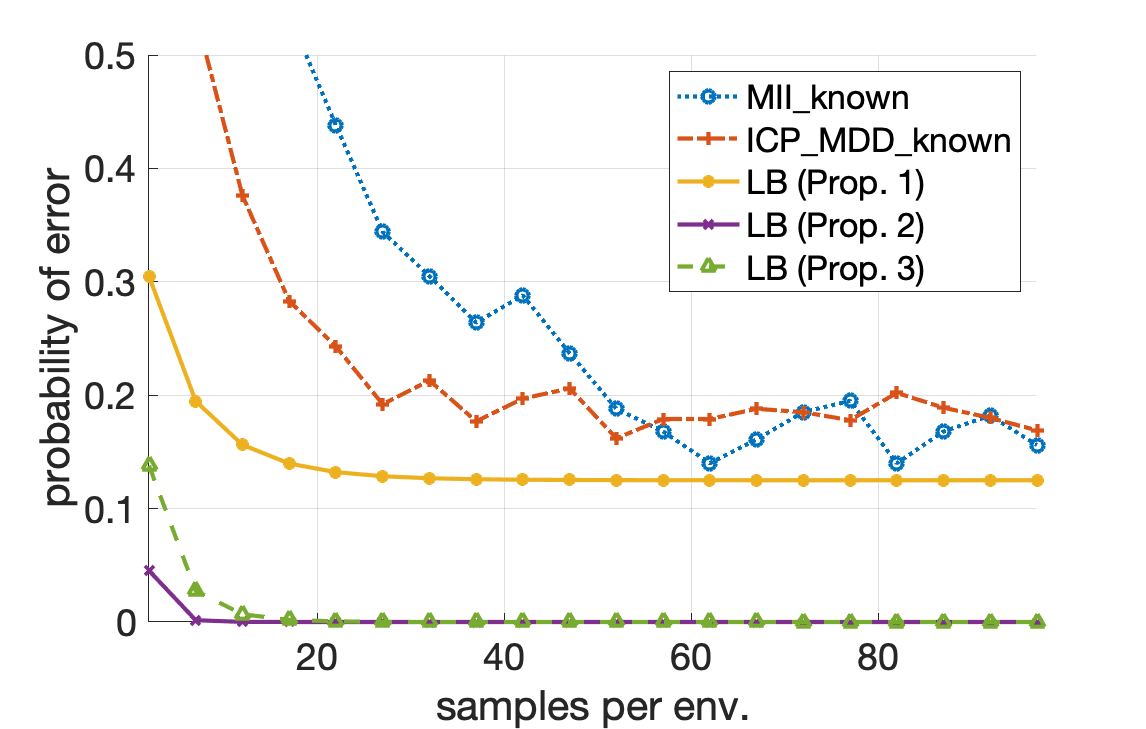}
    \caption{}
    \label{fig:simp_ss}
  \end{subfigure}
  \begin{subfigure}[b]{0.333\textwidth}    \includegraphics[width=\linewidth]{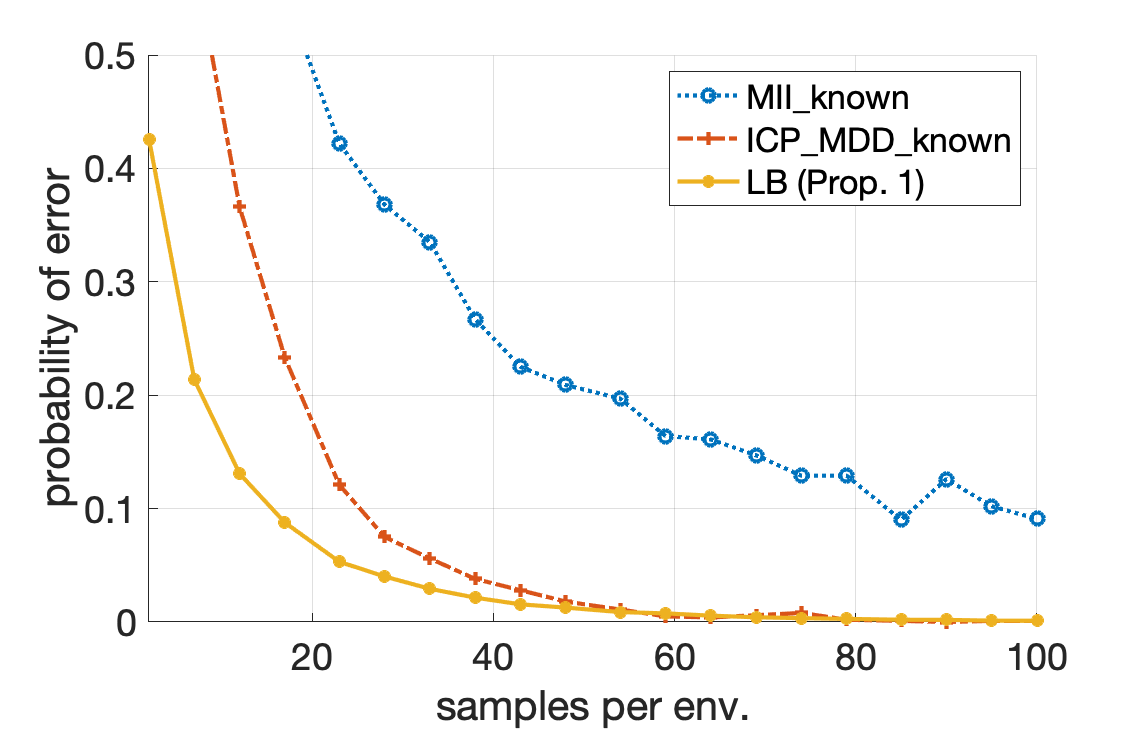}
    \caption{}
    \label{fig:known_ss}
  \end{subfigure}
  \begin{subfigure}[b]{0.333\textwidth}
    \includegraphics[width=\linewidth]{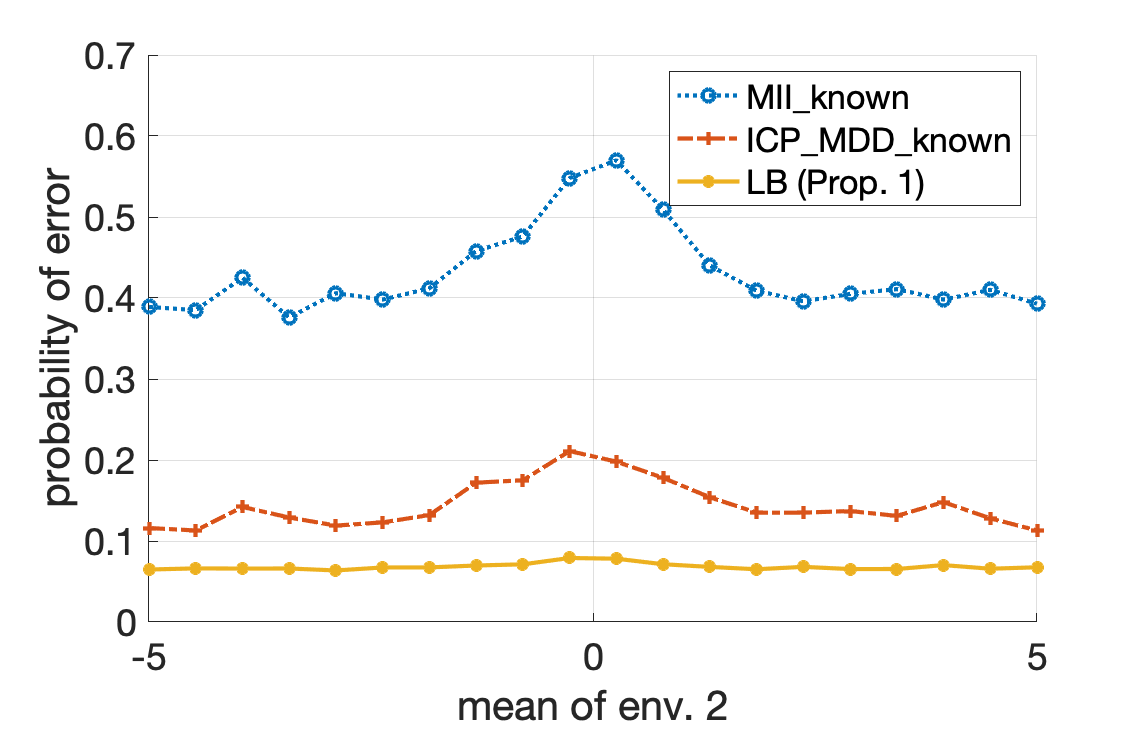}
    \caption{}
    \label{fig:known_mean}
  \end{subfigure}
  \caption{ Comparison of lower bounds: (a) for each environment, predictors are simplex codes within a sphere of radius $\sqrt{n_e}$, (b) predictors simulated Gaussian \textsf{SEM}s with a random number of edges between parents, and (c) predictors are simulated Gaussian \textsf{SEM}s where interventions constitute shifts in the mean of the second environment.}
  \vspace{-1.5em}
\end{figure*}
\end{center}

\vspace{-2em}
\noindent\underline{\bf Random Gaussian structural equation models.}
Next, each simulated dataset is a random Gaussian structural equation model (\textsf{SEM})~\cite{hox1998introduction}. Coefficients for the \textsf{SEM} are all randomly selected between $0.5$ and $1.5$. All noise variables are distributed according to $\mathcal{N}(0,1)$.  The value of $m$ is a random integer on $[3,\dots,8]$. The set of causal parents $S^*$ is chosen uniformly from $\Tc_m$. Unless otherwise specified, the number of edges between predictors is a random integer on the interval $[0,\dots,{m \choose 2}]$. We consider interventions that shift the mean of the top-level variables in the graph (i.e., those with no parents). 

Two settings in this regime are examined. The first fixes the intervention means to zero for environment one and one for environment two while increasing the number of samples per environment. While the methods \textsf{ICP\_MDD\_known} and \textsf{MII\_known} both approach the bound in Proposition~\ref{cor:lb1}, \textsf{ICP\_MDD\_known} does so at a faster rate than that of \textsf{MII\_known} (see Figure~\ref{fig:known_ss}). As these random codes represent a more challenging setting compared to the previously discussed simplex codes, the bounds in Propositions~\ref{cor:lb2} and~\ref{cor:lb3} are close to zero, and are thus not included in the figures. 

In the next setting, we vary the intervention mean for environment two from $-5$ to $5$.  The sample size per environment is $20$. Again, we find the error probability of \textsf{ICP\_MDD\_known} is closer to the bound in Proposition~\ref{cor:lb1} than that of \textsf{MII\_known}. We observe spikes in the error probability when the means of both environments are $0$. This is expected as $S = \{\}$ (the empty set) will likely be incorrectly decoded when codewords are close to zero. The lower bound in Proposition~\ref{cor:lb1} remains mostly invariant to shifts in the mean (Figure~\ref{fig:known_mean}). 

\subsection{Unknown coefficient settings}\label{ssec:unknown}
We now compare the proposed lower bounds to the more applicable setting where the coefficients $w$ are unknown. As estimating coefficients will only cause the probability of error to increase, we assert that the bounds in Propositions~\ref{cor:lb1},~\ref{cor:lb2}, and~\ref{cor:lb3} apply to this more general setting as well. 

\begin{algorithm}
\label{alg:unknown}
\caption{\textsf{ICP\_MDD}}\label{alg::ICP_MDD_UNK}
\hspace*{\algorithmicindent} \textbf{Input:} Response $Y^e$, predictors $X^e$, and threshold $p$

\hspace*{\algorithmicindent}
\textbf{Output:} The estimate $\hat{S}^*$
\begin{algorithmic}
\State Fit a linear regression model on pooled data $Y|X_S$ to obtain an estimate $\hat{\gamma}_S$ for every $S\in\Tc_m$.
\For{every $e\in\Ec$}
  \For{every $S\in\Tc_m$}
    \State Calculate $d^e_S = ||Y^e- X_S^e\hat{\gamma}_S||_2$
  \EndFor
  \State $d^e = \min_{S\in\Tc_m} d^e_S$
  \State $t^e = d^e + p*d^e$
\EndFor

\If{for all $S\in\Tc_m$, $d^e_S \leq t^e$ for all $e\in\Ec$} 
    \State accept subset $S$
\EndIf 

\State \textbf{return } intersection of all accepted sets

\end{algorithmic}
\end{algorithm}

Additionally, we propose the method \textsf{ICP\_MDD} as an extension to \textsf{ICP\_MDD\_known} for unknown coefficients and outline it in Algorithm~\ref{alg::ICP_MDD_UNK}. Perhaps the most important consideration now that $w$ is unknown is that $S^*$ is no longer unique. Because of this, we take a similar approach as in Method II in~\cite{peters2016causal}. Specifically, any subset that is a plausible $S^*$, we ``accept''. The estimate $\Sh^*$ is then the intersection of all accepted subsets. As we now potentially accept many subsets, we must define criteria for which a subset will be accepted. If, for each environment, the distance between the output $Y^e$ and the estimate $X_S^e\hat{\gamma}_S$ falls within some percentage $p$ of the minimum distance between output and estimate, then we accept that subset. Intuitively, one would want to choose the parameter $p$ to be small (e.g., $p = 0.1,0.05,0.01$). Additionally, we estimate coefficients for a given subset using data pooled over all environments. We refer to $X_S$ as the pooled dataset over all environments for some subset $S$. 

For the following experiments, we use the same random Gaussian \textsf{SEM} setup previously used and compare it with two other causal discovery methods. The first is Method II from~\cite{peters2016causal}. As we no longer assume the coefficients are known, we use the exact method proposed in~\cite{peters2016causal}. The second is the \textsf{LiNGAM}; in particular, we use the independent component analysis (\textsf{ICA}) based \textsf{LiNGAM} originally proposed in~\cite{shimizu2006linear}. 

We examine two settings for the unknown coefficient setting. We first simulate \textsf{SEM}s such that each predictor is generated independently of all others. Shifts in the intervention mean for environments one and two are fixed to zero and one, respectively. Results indicate \textsf{LiNGAM}, \textsf{ICP\_MDD}, and Method II perform comparably with \textsf{ICP\_MDD} performing slightly better for smaller sample sizes and \textsf{LiNGAM} performing slightly better for larger sample sizes (Figure~\ref{fig:unk_ss_indep}). 

The next setting incorporates random edges between predictors as was done in the known coefficient experiments. With this addition, the error probability of all methods decreases. Apart from small sample sizes less than $25$, \textsf{LiNGAM} achieves the lowest probability of error, followed by \textsf{ICP\_MDD}, then Method II (Figure~\ref{fig:unk_ss}). Bounds in Propositions~\ref{cor:lb2} and~\ref{cor:lb3} are near zero, so we omit them in Figures~\ref{fig:unk_ss_indep} and~\ref{fig:unk_ss}.

\begin{center}
\begin{figure}
  \begin{subfigure}[b]{0.49\columnwidth}    \includegraphics[width=\linewidth]{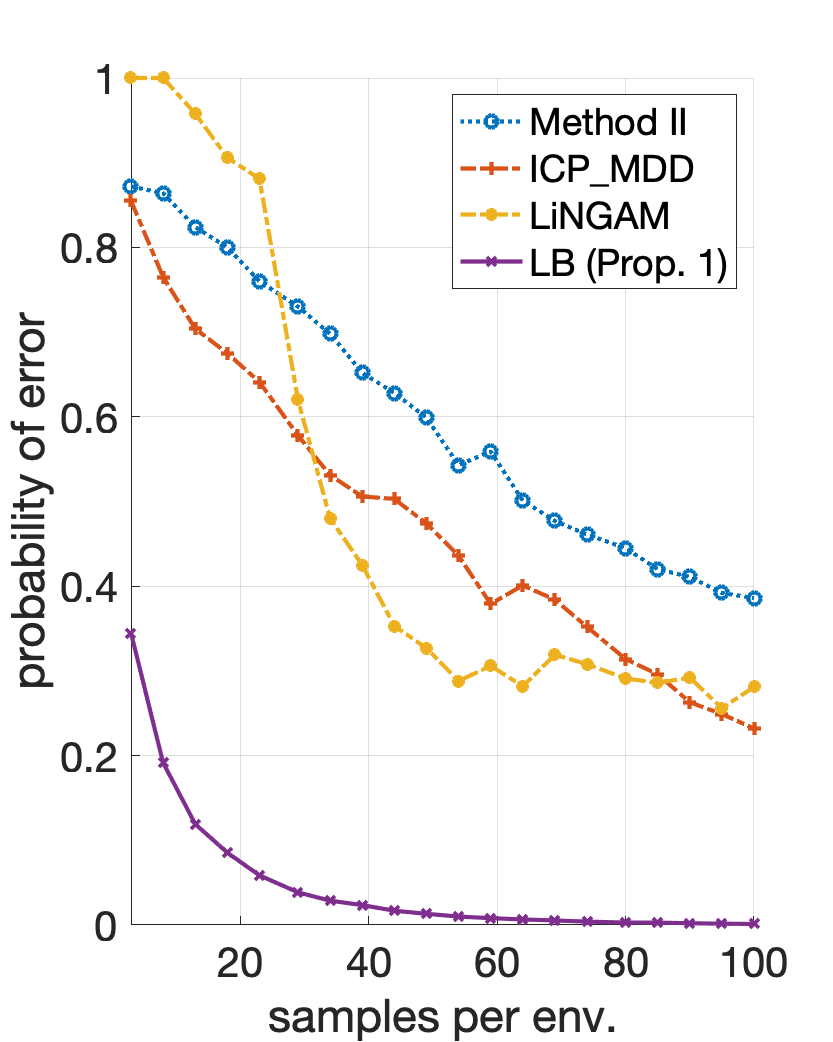}
    \caption{}
    \label{fig:unk_ss_indep}
  \end{subfigure}
  \begin{subfigure}[b]{0.49\columnwidth}
    \includegraphics[width=\linewidth]{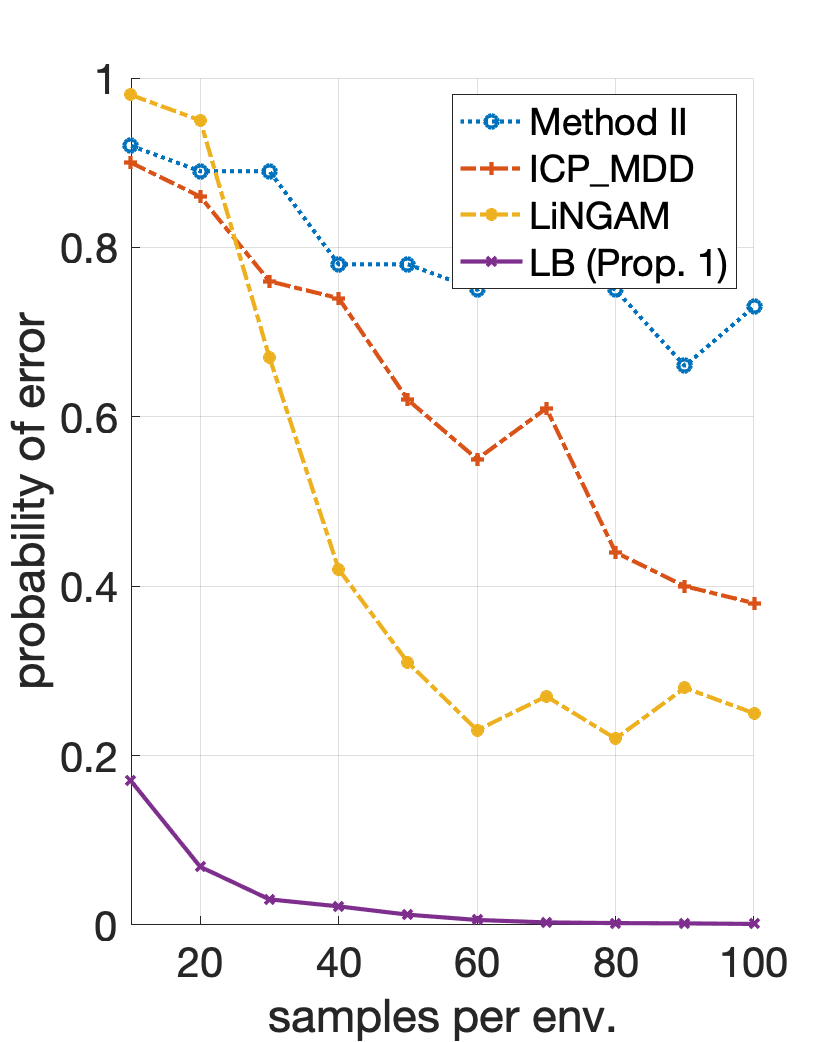}
    \caption{}
    \label{fig:unk_ss}
  \end{subfigure}
  \caption{Unknown coefficient and Gaussian \textsf{SEM} setup: (a) all predictors are sampled independently, and (b) there are a random number of edges between predictors.}
  \vspace{-1.5em}
\end{figure}
\end{center}

\vspace{-2.5em}
\balance
\bibliographystyle{IEEEtran}
\bibliography{sample}

\clearpage
\appendix
\subsection{Proof of Proposition~\ref{cor:lb2}}
For convenience, we will temporarily drop the superscript $e$ in the equations below. The average squared distance $\Bar{d}^2$ over all combinations of sent codewords is
\begin{align*}
&\frac{1}{2 {|\Tc_m| \choose 2}} \sum_{S_1 \in \Tc_m} \sum_{S_2 \in \Tc_m} d^2_{S_1,S_2} \\
 &=  \frac{1}{2 {2^m \choose 2}} \sum_{t=1}^{n_e} \sum_{S_1 \in \Tc_m} \sum_{S_2 \in \Tc_m}  (v_{S_1,t} - v_{S_2,t})^2 \\
 &=  \frac{1}{2 {2^m \choose 2}} \sum_{t=1}^{n_e} \sum_{S_1 \in \Tc_m} \sum_{S_2 \in \Tc_m}  (v^2_{S_1,t} + v^2_{S_2,t} - 2v_{S_1,t}v_{S_2,t})\\
 &=  \frac{1}{{2^m \choose 2}} \sum_{t=1}^{n_e} \Biggl[ 2^m \sum_{S_1 \in \Tc_m} v^2_{S_1,t}-\sum_{S_1 \in \Tc_m} \sum_{S_2 \in \Tc_m}v_{S_1,t}v_{S_2,t}\Biggl]. 
 \numberthis \label{eqn:g_avgd_1}
\end{align*}
We now treat each of the terms in \eqref{eqn:g_avgd_1} individually. Let $\Tc_m^k$ be the set of all subsets of $ \{1,\dots,m \}$ length $k$. The term $\sum_{S_1 \in \Tc_m} v^2_{S_1,t}$ becomes 
\begin{align*}
 &\sum_{S_1 \in \mathcal{T}_m} \left( \sum_{i \in S_1} w_ix_{i,t}\right)^2 \\
 &= \sum_{S_1 \in \Tc_m}\sum_{i \in S_1}\sum_{\substack{j\in S_1 \\ i\neq j}} w_iw_jx_{i,t}x_{j,t} + \sum_{S_1 \in \Tc_m}\sum_{i \in S_1} w_i^2x_{i,t}^2 \\
 &= \sum_{k=0}^m\sum_{S_1 \in \Tc_m^k}\sum_{i \in S_1}\sum_{\substack{j\in S_1 \\ i\neq j}} w_iw_jx_{i,t}x_{j,t} + \sum_{k=0}^m\sum_{S_1 \in \Tc_m^k}\sum_{i \in S_1} w_i^2x_{i,t}^2 \numberthis \label{eqn:g_term1_1}
 \end{align*}
 The expression in~\eqref{eqn:g_term1_1} is further simplified using two properties. Namely, for any $a = [a_1,\dots, a_m]\in \mathbb{R}^{m}$, we have 
 \begin{equation}\label{equ:binom1}
 \sum_{S_1\in \Tc^k_m}\sum_{i\in S_1} a_i = {m-1\choose k-1}\sum_{i=1}^m a_i,
 \end{equation}
 \begin{equation}\label{eqn:binom2}
     \sum_{S_1\in \Tc^k_m} \sum_{i\in S_1} \sum_{\substack{j\in S_1 \\ i\neq j}} a_ia_j = 2{m-2\choose k-2} \sum_{i=1}^{m-1} \sum_{j=i+1}^m a_ia_j.
 \end{equation}
 We derive~\eqref{equ:binom1} by counting the occurrences of each  $a_i$ for $i$ from $\{1,\dots,m\}$. Similarly,~\eqref{eqn:binom2} can be derived by counting each pair $(a_i,a_j)$ for $i,j$ from $\{1,\dots,m\}$, where $i\ne j$. Thus,~\eqref{eqn:g_term1_1} becomes 
 \begin{align*}
 & 2\sum_{k=0}^m {m-2 \choose k-2}\sum_{i=1}^{m-1}\sum_{j=i+1}^m w_iw_jx_{i,t}x_{j,t} \\
 &\qquad + \sum_{k=0}^m{m-1 \choose k-1}\sum_{i =1}^m w_i^2x_{i,t}^2 \\
 &= \sum_{k=0}^m\frac{k(k-1)}{m(m-1)}{m \choose k}\left[ \sum_{i=1}^{m}\sum_{j=1}^m w_iw_jx_{i,t}x_{j,t} - \sum_{i=1}^m w_i^2x_{i,t}^2\right] \\
 &\qquad + \sum_{k=0}^m\frac{k}{m}{m \choose k}\sum_{i =1}^m w_i^2x_{i,t}^2 \\
 &\leq 2^{m-2}\left(\sum_{i=1}^{m} w_ix_{i,t} \right)^2  
+ 2^{m-1}\sum_{i =1}^m w_i^2x_{i,t}^2, \numberthis \label{eqn:g_term1}
\end{align*}
where the first term in ~\eqref{eqn:g_term1} follows since $\sum_{k=0}^m k{m \choose k} = m2^{m-1}$ and the second term follows since $\sum_{k=0}^m k^2{m \choose k} = m(m+1)2^{m-2}$. Now returning to~\eqref{eqn:g_avgd_1}, the second term $\sum_{S_1 \in \Tc_m}\sum_{S_2 \in \Tc_m}v_{S_1,t}v_{S_2,t}$ becomes
\begin{align*}
 \sum_{S_1 \in \Tc_m}&\sum_{S_2 \in \Tc_m} \left( \sum_{i\in S_1} w_ix_{i,t} \right) \left( \sum_{j\in S_2} w_jx_{j,t} \right) \\
  & = \left(\sum_{k=0}^m\sum_{S_1 \in \Tc_m^k}\sum_{i\in S_1}  w_ix_{i,t}\right)^2 \\
  & =\left(\sum_{k=0}^m{m-1 \choose k-1}\right)^2 \left(\sum_{i = 1}^m w_ix_{i,t}\right)^2 \\
 & = 2^{2m-2} \left(\sum_{i = 1}^m w_ix_{i,t}\right)^2,\numberthis \label{eqn:g_term2}
\end{align*}
where~\eqref{eqn:g_term2} follows using the same arguments as in~\eqref{eqn:g_term1}. From \eqref{eqn:g_term1} and \eqref{eqn:g_term2}, the two inner terms in \eqref{eqn:g_avgd_1} become
\begin{align*}
 \sum_{t=1}^{n_e} &\Biggl[ 2^m \sum_{S_1 \in \Tc_m} v^2_{S_1,t} -\sum_{S_1 \in \Tc_m} \sum_{S_2 \in \Tc_m}v_{S_1,t}v_{S_2,t}\Biggl] \\
  & \leq 2^{2m-1}\sum_{t=1}^{n_e}\sum_{i =1}^m w_i^2x_{i,t}^2 \\
  &\quad + \sum_{t=1}^{n_e}\Biggr[ 2^{2m-2}\left(\sum_{i=1}^{m} w_ix_{i,t} \right)^2 - 2^{2m-2} \left(\sum_{i = 1}^m w_ix_{i,s}\right)^2\Biggr] \\
  & = 2^{2m-1}\sum_{t=1}^{n_e}\sum_{i =1}^m w_i^2x_{i,t}^2  \\
&\leq 2^{2m-1}mn_eP_e,   \numberthis \label{eqn:g_terms_comb}
\end{align*}
where~\eqref{eqn:g_terms_comb} follows from the power constraint in~\eqref{eqn::peak_pwr_2}. From \eqref{eqn:g_terms_comb}, it follows that the average squared distance is, 
\begin{align*} 
  \Bar{d}^2  &\leq \frac{2^{2m-1}mn_eP_e}{{2^m \choose 2}} 
  = \frac{2^{m} mn_eP_e}{2^m-1}. \numberthis \label{eqn:avg_sqr_dst}
\end{align*}
 Since there is at least one pair of signals for which~\eqref{eqn:avg_sqr_dst} holds, we can bound $\Ps\{ \hat{S}^*_e \ne S^* | S^* = S\}$ using the average distance as opposed to the actual distance. Specifically, since the noise is Gaussian, $\Ps\{ \hat{S}^*_e \ne S^* | S^* = S\}$ constitutes the probability that the noise contribution
moves $v_S$ at least half the distance between it and the next closest possible sent signal. i.e., there exists some $S\in\Tc_m$ such that 
\begin{align*}
    \Ps\{ \hat{S}^*_e \ne S^* | S^* = S\} &\ge \Phi\left(-\frac{\Bar{d}^2}{2\s_\text{min}} \right)\\
    &\ge \Phi\left( -\sqrt{\frac{2^{m}mn_eP_e}{4\s^2_\text{min}\left(2^m-1\right)}} \right). \numberthis \label{eqn:P2_perr_cond}
\end{align*}
This accounts for only one of the $|\Tc_m|$ contributions to the total error. By removing some of the $|\Tc_m|$ combinations from the list of sendable signals, other contributions to the error can be identified.  We remove an entire codeword from the codebook, which accounts for $2^{m-1}$ of the total combinations. Since removing combinations will cause the bound on the average distance in~\eqref{eqn:avg_sqr_dst} to increase and the bound probability of error in~\eqref{eqn:P2_perr_cond} to decrease, the first $2^{m-1}$ combinations can be bounded by the probability of the $(2^{m-1})$-th combination. That is, the first $2^{m-1}$ combinations contribute at least
\begin{equation*}
2^{m-1} \Phi\left( - \sqrt{\frac{2^{m-1}(m-1)n_eP_e}{4\s^2_\text{min}\left(2^{m-1}-1\right)}} \right)
\end{equation*}
to the probability of error proportional to the remaining contributions. Another codeword can be removed, accounting for $2^{m-2}$ combinations and a bound on the contribution to the probability of error can be derived. This process is repeated by removing codewords and adding contributions until all but two codewords have been removed. We then obtain the following bound 
\begin{align*}
\Ps\{ \hat{S}^*_e \ne S^* \} &\ge \frac{1}{2^m}\sum_{i = 1}^{m-1} 2^{m-i} \Phi\left( - \sqrt{\frac{2^{m-i}(m-i)n_eP_e}{4\s^2_\text{min}\left(2^{m-i}-1\right)}} \right) \\
&= \sum_{i = 1}^{m-1} \frac{1}{2^i} \Phi\left( - \sqrt{\frac{2^{m-i}(m-i)n_eP_e}{4\s^2_\text{min}\left(2^{m-i}-1\right)}} \right). \numberthis \label{eqn:lb2_mac_only}
\end{align*}
The probability of error in~\eqref{eqn::perr1} is then bounded by
\begin{equation*} \label{eqn:bnd2_proof}
    P_{\text{err}} \ge \max_{e\in\Ec}\sum_{i = 1}^{m-1}\frac{1}{2^i}  \Phi\left( - \sqrt{\frac{2^{m-i}(m-i)n_eP_e}{4\s^2_{\text{min}}\left(2^{m-i}-1\right)}} \right).
    \end{equation*}

\end{document}

%% file: defns.tex
\usepackage{xspace}
\usepackage{bbm}
\usepackage{mathrsfs}

\newcommand{\Ps}{{\sf P}}

\newcommand{\Ec}{\mathcal{E}}

\newcommand{\Tc}{\mathcal{T}}

\newcommand{\Sh}{{\hat{S}}}

\def\g{\gamma}

\def\s{\sigma}

\def\textiid{i.i.d.\@\xspace}
\newcommand\iid{\ifmmode\text{ i.i.d. } \else \textiid \fi}

